\documentclass[hidelinks]{ifacconf}
\usepackage{algorithm}
\usepackage{algpseudocode}

\usepackage{tabularx}
\usepackage{makecell} 
\usepackage{array} 
\usepackage{xcolor}
\usepackage{tabularx,makecell,array}
\usepackage{enumitem} 
\usepackage{amssymb}
\usepackage{amsfonts}
\usepackage{amsmath}
\usepackage{graphicx}     
\usepackage{natbib}        
\usepackage{doi}

\makeatletter
\let\bibliographystyle\@gobble
\AtBeginDocument{%
  \renewenvironment{thebibliography}[1]{%
    \par\addvspace{6pt}%
    \noindent{\bfseries \refname}\par\nobreak\addvspace{3pt}%
    \small
    \list{}{%
      \leftmargin=1em%
      \itemindent=-1em%
      \labelwidth=0pt%
      \labelsep=0pt%
      \itemsep=0pt%
      \parsep=0pt%
    }%
    \sloppy
    \clubpenalty4000
    \widowpenalty4000
    \sfcode`\.=\@m
    \let\NAT@bibitem@first@sw\@firstoftwo
  }{\endlist}%
}
\makeatother

\makeatletter
\def\doi#1{}
\makeatother


\usepackage{enumitem}
\usepackage{stfloats}   
\newcommand{\hop}{\mathsf{H}}
\DeclareMathOperator*{\argmin}{argmin}

\newtheorem{theorem}{Theorem} 
\newtheorem{remark}{Remark} 
\newtheorem{definition}{Definition} 

\newtheorem{corollary}{Corollary}
\newtheorem{example}{Example}

\newtheorem{problem}{Problem}
\newenvironment{proof}{\begin{pf}}{\qed\end{pf}}

\usepackage{xcolor}
\usepackage{tikz}
\definecolor{dblue}{rgb}{0 0 0.7}
\definecolor{red}{rgb}{1 0 0}
\definecolor{MatlabBlue}{rgb}{0     , 0.4470, 0.7410}
\definecolor{MatlabRed}{rgb}{0.6350 0.0780 0.1840}
\definecolor{MatlabOrange}{rgb}{0.8500, 0.3250, 0.0980}
\definecolor{MatlabYellow}{rgb}{0.9290, 0.6940, 0.1250}
\definecolor{MatlabPurple}{rgb}{0.4940, 0.1840, 0.5560}
\definecolor{MatlabGreen}{rgb}{0.4660, 0.6740, 0.1880}
\definecolor{MatlabBabyBlue}{rgb}{0.3010, 0.7450, 0.9330}
\definecolor{MatlabGray}{rgb}{0.5, 0.5, 0.5}
\definecolor{MatlabLightGray}{rgb}{0.75, 0.75, 0.75}
\definecolor{MatlabBlack}{rgb}{0, 0, 0}
\definecolor{MatlabLightGray4}{rgb}{0.875, 0.875, 0.875}
\definecolor{MatlabLightGray3}{rgb}{0.85, 0.85, 0.85}
\definecolor{MatlabLightGray2}{rgb}{0.775, 0.775, 0.775}
\definecolor{MatlabLightGray1}{rgb}{0.7, 0.7, 0.7}
\definecolor{MatlabGray20}{rgb}{0.2, 0.2, 0.2}
\definecolor{MatlabGray30}{rgb}{0.3, 0.3, 0.3}
\definecolor{MatlabGray40}{rgb}{0.4, 0.4, 0.4}
\definecolor{MatlabGray50}{rgb}{0.5, 0.5, 0.5}
\definecolor{MatlabGray60}{rgb}{0.6, 0.6, 0.6}
\definecolor{MatlabGray70}{rgb}{0.7, 0.7, 0.7}
\definecolor{MatlabGray80}{rgb}{0.8, 0.8, 0.8}
\definecolor{MatlabGray85}{rgb}{0.85, 0.85, 0.85}
\definecolor{MatlabGray90}{rgb}{0.9, 0.9, 0.9}
 
\newcommand{\tikzline}[1]{(\protect\tikz[baseline=-0.6ex,x=1pt,y=1pt, line width=0.4mm]{ \protect\draw[#1] [-] (0,0) -- (10,0);})}
\newcommand{\tikzdashedline}[1]{(\protect\tikz[baseline=-0.6ex,x=0.9pt,y=1pt]{ \protect\draw[#1,dashed] [-] (0,0) -- (10,0);})}

\newcommand{\tikzdot}[1]{%
(\protect\tikz[baseline=-0.6ex,x=1pt,y=1pt]{%
    \protect\draw[#1, fill=#1] (0,0) circle (2pt);})%
}

\usepackage{tikz} 

\makeatletter
\renewcommand\thesubsubsection{\arabic{section}.\arabic{subsection}.\arabic{subsubsection}}
\renewcommand\subsubsection{\@startsection{subsubsection}{3}{\z@}%
  {1.0ex \@plus .2ex}%
  {0.5ex}%
  {\normalfont\normalsize\bfseries}}
\makeatother

\makeatletter
\providecommand{\addperiod}[1]{#1}
\renewcommand{\addperiod}[1]{%
  \ifnum\pdfstrcmp{\expandafter\@gobble\string#1}{.}=0\relax
  \else.%
  \fi}
\makeatother

\usepackage{eso-pic}
\AddToShipoutPictureBG*{%
  \AtPageUpperLeft{%
    \setlength\unitlength{1in}%
    \hspace*{\dimexpr0.5\paperwidth\relax}%
    \makebox(0,-2.9)[c]{%
      \begin{tabular}{c}
        Maarten van der Hulst \emph{et al.}, \\
        \emph{Sparse add-on controller design: A Youla approach to system-level performance}, \\
        uploaded to arXiv \today \\
      \end{tabular}%
    }%
  }%
}

\begin{document}
\begin{frontmatter}

\title{Sparse add-on controller design: A Youla approach to system-level performance}

\thanks[footnoteinfo]{This project is funded by Holland High Tech | TKI HTSM via the PPP Innovation Scheme (PPP-I) for public-private partnerships.}

\author[First]{M. van der Hulst}
\author[Second]{N. Dirkx}
\author[First]{R. A. Gonz\'alez}
\author[First]{K. Tiels}
\author[Second]{J. van de Wijdeven}
\author[First,Third]{T. Oomen}

\address[First]{Dept. of Mechanical Engineering, Eindhoven University of Technology, The Netherlands}
\address[Second]{ASML, Veldhoven, The Netherlands}
\address[Third]{Delft Center for Systems and Control, Delft University of Technology, The Netherlands}

\begin{abstract}             
The performance of high-tech systems is often dictated by a few performance objectives shared among the many closed-loop controlled subsystems operating in the machine, such as synchronization, coordination, and alignment, which necessitates control methods that explicitly address them to achieve optimal performance. The aim of this paper is to introduce a framework that improves system performance through system-level controllers designed to be implemented as add-ons to the existing subsystem control structure. The developed method parametrizes all stabilizing system-level add-on controllers using the Youla framework, enabling a convex formulation of the sparse $\mathcal{H}_2$ synthesis problem. The result is a sparse add-on controller that achieves the optimal trade-off between combined performance and interconnection complexity, as demonstrated through numerical simulations.  
\end{abstract}


\end{frontmatter}

\section{Introduction}
High-tech systems consist of numerous closed-loop controlled subsystems that must operate cooperatively to achieve overall system-level performance. Examples include ground-based telescopes, where the primary mirror comprises hundreds of individual segments that must be precisely aligned to achieve a uniform optical surface, industrial printing systems that rely on coordinated motion of belt drives and print heads to ensure print quality, and wafer scanners that require multiple motion stages to move synchronously for accurate wafer exposures. For such systems, achieving optimal system-level performance requires control methods that explicitly address the combined objectives across subsystems. For industrial adoption, it is essential that these methods operate as add-on to the existing control structure and limit the number of interconnections to reduce communication links and implementation complexity, both of which become critical in large-scale systems.

The current industry standard for designing controllers for such systems is using decentralized control methods \citep{Skogestad2006MultivariableDesign, Oomen2018AdvancedSystems}, where subsystems are designed largely individually to meet local specifications that collectively aim to achieve overall system performance \citep{Heertjes2020ControlDevelopments}. While providing clear benefits  in terms of modularity and design simplicity, substantial design freedom remains unexploited in terms of inter-subsystem coupling, resulting in suboptimal system-level performance despite the locally optimal designs.

Early approaches aimed at improving system-level performance introduced coupling controllers as add-ons to the existing decentralized controllers \citep{Tay1997HighControl}. These include $\mathcal{H}_\infty$-based synthesis methods \citep{Looijen2018RobustOptimization}, iterative learning control techniques \citep{Barton2008AControl, Mishra2008IterativeStages}, and data-driven methods for online optimization \citep{Heertjes2013Self-tuningSystems}. A more general framework for bi-directional coupling control design based on Youla–Kučera parameterizations was introduced by \cite{Evers2019BeyondMotion}, enabling both $\mathcal{H}_2/\mathcal{H}_\infty$ synthesis and manual loop-shaping design. While these frameworks provide effective solutions, they are limited to two-subsystem cases, mainly addressing motion-stage synchronization problems.

Many $\mathcal{H}_2$-based structured synthesis methods for sparse controller design have been developed for large-scale interconnected systems. In \cite{Rotkowitz2005AControl, Matni2014RegularizationDesign} the property of Quadratic Invariance (QI) is exploited to enable convex synthesis of structured controllers, whereas \cite{Polyak2013AnSystems, Lin2013DesignMultipliers, Fardad2014OnProgramming} address the general non-QI case. More recently, \cite{Wang2019ASynthesis} introduced a system-level synthesis framework that generalizes structured controller design for a large class of convex performance objectives. However, these methods require a full redesign of the controller, which is often undesirable in high-tech applications where subsystems are designed independently and integrated only at a later stage, requiring the coupling to be implemented as add on.

Although important progress has been made in the control of interconnected 
systems, general methods that build directly on existing subsystem controllers 
to improve system-level performance remains lacking. This paper introduces 
a system-level add-on control framework for systems of multiple locally 
controlled subsystems. The main contributions are:
\begin{itemize}
    \item[C1] A Youla-based system-level add-on controller parameterization 
    for systems of multiple locally controlled subsystems, enabling 
    system-level performance improvement while keeping the existing 
    local controllers fixed.
    \item[C2] A convex $\mathcal{H}_2$ synthesis procedure for sparse 
    system-level controllers, enabling an optimal trade-off between 
    system-level performance and interconnection complexity.
\end{itemize}
The parameterization in C1 generalizes the two-subsystem Youla-based 
framework introduced in \cite{Evers2019BeyondMotion} to systems of 
arbitrary size. This paper is organized as follows. Section~II 
introduces notation and preliminaries. Section~III formulates the 
add-on control problem. Section~IV develops the system-level controller 
parameterization and performance analysis. Section~V presents the 
synthesis procedure. Section~VI provides simulation results, and 
Section~VII presents concluding remarks.

\section{Notation \& preliminaries}  The set of real numbers is denoted by $\mathbb{R}$. For a matrix $  {A} $, its transpose is written as $  {A}^{\top} $, and its Hermitian (conjugate transpose) as $  {A}^{\hop} $. The Kronecker product is represented by $ \otimes $. For $  {X} = [ {x}_1, \ldots,  {x}_n] $ with $  {x}_i \in \mathbb{C}^n $ the operation $ \operatorname{vec}( {X}) = [ {x}^\top_1, \ldots,  {x}^\top_n]^\top $ restructures the matrix into a vector by stacking its columns. The null space of a matrix $A$, i.e., the set of all vectors $x$ satisfying $Ax = 0$, is denoted by $\operatorname{ker}(A)$. The set of real rational transfer matrices of dimension $n \times m$ is denoted by $\mathcal{R}^{n \times m}$, and $\mathcal{RH}_{\infty}^{n \times m}$ denotes the subset of stable and proper transfer matrices. The Laplace operator $s$ is omitted if clear from context.  The following definitions are used \citep{Zhou1999EssentialsControl}. 
\begin{definition}[$\mathcal{H}_2$ norm]
For a stable, strictly proper LTI system \(G\), the \(\mathcal{H}_2\) norm is
defined as
\[
\| G \|_{\mathcal{H}_2}
=
\sqrt{
\frac{1}{2\pi}
\int_{-\infty}^{\infty}
\operatorname{tr}\!\left(
G^{\hop}(j\omega)G(j\omega)
\right)d\omega
}.
\]
\end{definition}

\begin{definition}[Right-coprime factorization]
Let $G \in \mathcal{R}^{n \times m}$.  
The ordered pair $\{N, D\}$, with $D \in \mathcal{RH}_{\infty}^{m \times m}$ invertible over $\mathcal{R}^{m \times m}$ and $N \in \mathcal{RH}_{\infty}^{n \times m}$, is called a right-coprime factorization (RCF) of $G$ if $G = N D^{-1}$ and there exist $L\in \mathcal{RH}_\infty^{m \times n}$ and $ F\in \mathcal{RH}_\infty^{m \times m}$ such that the Bézout identity holds
\[
L N + F D = I.
\]
\end{definition}
\begin{definition}[Lower linear fractional transformation]
The lower linear fractional transformation (LFT) of a partitioned transfer matrix \(M = \begin{bmatrix} M_{11} & M_{12} \\ M_{21} & M_{22} \end{bmatrix}\) with respect to a transfer matrix \(K\) is defined as \(\mathcal{F}_\ell(M, K) = M_{11} + M_{12} K (I - M_{22} K)^{-1} M_{21}\). The LFT is well-posed if \(I - M_{22} K\) is invertible almost everywhere in $\mathbb{C}$.
\end{definition}

\section{Problem formulation}
This section introduces the control setting and formulates the coupling control design problem addressed in this paper.

\subsection{Baseline system description}
Consider the baseline system \(\Sigma ^0\) consisting of \(N\) locally controlled subsystems, described by
\begin{equation} \label{eq: baseline S}
\Sigma^0:\
\left\{
\begin{aligned}
    y &= G(u + d_u) + d_y, \\
    e &= r-y-n, \\
    u &= K_0 e,
\end{aligned}
\right.
\end{equation}
where \(u \in \mathbb{R}^N\) denotes the control input, \(y \in \mathbb{R}^N\) the 
measured output, \(r \in \mathbb{R}^N\) the reference signal, input disturbance \(d_u \in \mathbb{R}^N\), output disturbance \(d_y \in \mathbb{R}^N\), \(n \in \mathbb{R}^N\) the sensor noise, and \(e \in \mathbb{R}^N\) 
the local error signal. The plant \(G \in \mathcal{R}^{N \times N}\) is defined by
\[
G = \operatorname{diag}(G_{11},\dots,G_{NN}),
\]
where each \(G_{ii}\) represents the local subsystem dynamics and dynamic interactions between subsystems are assumed negligible. The interconnected system is controlled by the known and fixed decentralized controller \( K_0 \in \mathcal{RH}_{\infty}^{N \times N} \), given by
\[
K_0 = \operatorname{diag}(K_{11},\dots,K_{NN}),
\]
where \( K_{ii} \) denotes the local controller of subsystem \( i \), designed to stabilize the closed-loop system and achieve local performance specifications. The disturbance \(d = [d^\top_y, d^\top_u]^\top\) is generated as \(d = Vd_{\mathcal S}\), where \(V \in \mathcal{RH}_{\infty}^{2N \times 2N}\) denotes a full normal-rank disturbance model and $d_{\mathcal{S}}\in \mathbb{R}^{2N}$ the system-level disturbance. The performance of the baseline system is characterized by the system-level error signal \(e_{\mathcal S} \in \mathbb{R}^{n_w}\), defined as
\begin{equation}
    e_{\mathcal S} = W e,
\end{equation}
where \(W \in \mathbb{R}^{n_w \times N}\), with \(\operatorname{rank}(W)=r_w<N\), is a static matrix that encodes the system-level objective across the local subsystems. Since \(W\) has reduced rank, the system-level error $e_{\mathcal{S}}$ depends only on a lower-dimensional projection of the local error vector~\(e\). \\

\begin{remark}
    The signal definitions and dimensions are stated for SISO subsystems, but extend directly to the MIMO subsystems case. 
\end{remark}

\subsection{System-level add-on controller design problem}
To improve system-level performance, the baseline controller \(K_0\) is 
augmented with an add-on system-level controller, yielding the interconnected 
system \(\Sigma\). Since \(W\) has a nontrivial null space, system-level 
performance can be improved by redistributing the local error vector \(e\) 
to align with \(\ker(W)\), thereby reducing the component observed in 
\(e_{\mathcal S}\), beyond what is achievable by \(K_0\) alone. This leads 
to the following design problem. \\

\begin{problem} \label{prob:1}
Given the baseline system \(\Sigma^0\) with fixed local controller \(K_0\), 
design an add-on system-level controller that minimizes \(e_{\mathcal S}\), 
while satisfying:
\begin{enumerate}
\renewcommand{\labelenumi}{\arabic{enumi}.}
    \item Disabling the add-on controller recovers \(\Sigma^0\).
    \item Internal stability of \(\Sigma^0\) is preserved in \(\Sigma\).
\end{enumerate}
\end{problem}

\section{System-level add-on control framework}
This section develops the system-level add-on controller design framework, 
thereby providing Contribution~C1.

\subsection{Youla--Ku\v{c}era parameterizations}
First, the local subsystems are represented in 
generalized plant form as
\begin{equation} \label{eq:standard_plant}
\mathcal{P}:\
\begin{bmatrix}
z_{\mathcal L}\\
e
\end{bmatrix}
=
\begin{bmatrix}
P_{11} & P_{12}\\
P_{21} & P_{22}
\end{bmatrix}
\begin{bmatrix}
w_{\mathcal{L}}\\
u
\end{bmatrix},
\end{equation}
where \(w_{\mathcal{L}}=[d^\top,r^\top,n^\top]^\top\) are the local exogenous 
inputs, $ z_{\mathcal L} = [e^\top, \bar z^\top_{\mathcal L}]^\top$, with \(\bar z_{\mathcal L}\) denoting additional local performance channels, and \(P_{11}\), \(P_{12}\), \(P_{21}\) are block diagonal with \(P_{22}=-G\). The local channels are mapped to 
system-level channels via \(z_{\mathcal S} = \mathcal W  z_{\mathcal L}\) and
\(w_{\mathcal L} = \mathcal V w_{\mathcal S}\), where
\begin{equation}
    \mathcal W=
    \begin{bmatrix}
        W & 0 \\
        0 & I
    \end{bmatrix},
    \quad
    \mathcal V =
    \begin{bmatrix}
        V & 0 \\
        0 & I
    \end{bmatrix}.
\end{equation}
The add-on controller is parameterized using the Youla--Ku\v{c}era 
framework, which characterizes all controller perturbations that internally stabilize the generalized plant $\mathcal{P}$ \citep{Zhou1999EssentialsControl}. Let $\{N_c, D_c\}$ be an RCF of the baseline controller $K_0$, and let $\{N_p, D_p\}$ be an RCF of a model of the true plant $\hat G$, where all coprime factors are diagonal. The standard Youla--Ku\v{c}era parameterization in LFT form is then given by
\begin{equation} \label{eq:youla_lft}
\begin{aligned}
\vcenter{\hbox{$\mathcal{K} :$}}\,
\begin{bmatrix}
u\\
\epsilon
\end{bmatrix}
&=
\begin{bmatrix}
N_cD_c^{-1} & -(D_p+N_cD_c^{-1}N_p)\\
D_c^{-1} & -D_c^{-1}N_p
\end{bmatrix}
\begin{bmatrix}
e\\
\xi
\end{bmatrix}, \\
\xi &= X\epsilon,
\end{aligned}
\end{equation}
where \(\epsilon \in \mathbb{R}^N\) and \(\xi \in \mathbb{R}^N\) denote the internal coupling signals that are physically exchanged between subsystems. 
The Youla parameter \(X\) is restricted to a static matrix \(X \in \mathbb{R}^{N\times N}\), which determines the controller interconnection structure. Specifically, \(X\) is structured according to
\begin{equation} \label{eq:X}
X =
\begin{bmatrix}
0      & X_{12} & \cdots & X_{1N} \\
X_{21} & 0      & \cdots & \vdots \\
\vdots & \vdots & \ddots & X_{N-1,N} \\
X_{N1} & \cdots & X_{N,N-1} & 0
\end{bmatrix},
\end{equation}
where each nonzero entry \(X_{ij}\) represents an interconnection from subsystem \(j\) to subsystem \(i\). Evaluating the LFT interconnection 
\(\bar K(X)=\mathcal{F}_{\ell}(\mathcal{K},X)\) yields a controller of the form
\[
\bar K(X)=K_0+K_\Delta(X),
\]
where \(K_\Delta(X)\) denotes the add-on controller induced by the 
Youla parameter \(X\). Hence, \(\bar K(X)\) consists of the fixed baseline 
controller \(K_0\) augmented with a coupling term of which the interconnection 
structure is encoded by \(X\). The corresponding LFT interconnection is shown 
in Figure~\ref{fig:lft}.

\begin{figure}[t]
    \centering
    \includegraphics[scale=1]{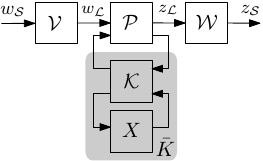}
\caption{LFT representation of the interconnected system with generalized plant \(\mathcal P\) and controller \(\bar K(X)\), where \(X\) encodes the interconnection structure between subsystems.}
    \label{fig:lft}
\end{figure}

\subsection{System-level add-on controller parameterization}
The generic Youla--Ku\v{c}era representation in~\eqref{eq:youla_lft} is 
used to derive an add-on controller parameterization by 
selecting suitable RCFs for \(K_0\) and the plant model \(\hat G\).  Since the local controller \(K_0\) is stable, a natural RCF is given by
\begin{equation} \label{eq:coprime_K}
\{N_c, D_c\} = \{K_0, I\}.
\end{equation}
The plant model \(\hat G\) may be unstable or contain non minimum-phase zeros. Therefore, the RCF is selected as
\begin{equation}  \label{eq:coprime_G}
\{N_p, D_p\} = \{Z, \hat G^{-1} Z\},
\end{equation}
where \(Z \in \mathcal{RH}_{\infty}^{N \times N}\) is a diagonal transfer function matrix which is chosen such that \(\hat G^{-1} Z \in \mathcal{RH}_{\infty}^{N \times N}\). In particular, \(Z\) collects the non minimum-phase zeros of \(\hat G\), and may include additional stable poles to ensure $\hat G^{-1}Z$ is proper. The following theorem presents the resulting add-on system-level controller parameterization, establishing the main result of this section.  \\

\begin{theorem}[Nominal add-on controller parameterization]\label{thm:theorem_1}
Consider the generalized plant \(\mathcal P\) in~\eqref{eq:standard_plant}, with 
RCFs given by~\eqref{eq:coprime_K} and~\eqref{eq:coprime_G}, and \(X\) as 
in~\eqref{eq:X}, and assume that $\hat G = G$. Then, the set of stabilizing system-level add-on controllers are parameterized by
\begin{equation} \label{eq: lft controller}
\begin{aligned}
\mathcal K:
\begin{bmatrix}
u\\
\epsilon
\end{bmatrix}
&=
\begin{bmatrix}
K_0 & -(S\hat G)^{-1}Z\\
I & -Z
\end{bmatrix}
\begin{bmatrix}
e\\
\xi
\end{bmatrix},\\
\xi&=X\epsilon,
\end{aligned}
\end{equation}
where \(S=(I+\hat G K_0)^{-1}\). The resulting controller \(\bar K(X)=\mathcal F_\ell(\mathcal K,X)\) satisfies
\[
\bar K(X) = K_0 - (S\hat G)^{-1}ZX(I+ZX)^{-1},
\]
and yields the closed-loop interconnection \(\mathcal{F}_\ell(\mathcal{P}, \bar K(X))\) with the following properties
\begin{enumerate}[label=(\roman*)] 
    \item \(\mathcal{F}_\ell(\mathcal{P}, \bar K(X))\) is affine in~\(X\).
    \item \(\mathcal{F}_\ell(\mathcal{P}, \bar K(X))\) is well-posed and internally stable.
    \item Each block $[\mathcal F_{\ell}(\mathcal P,\bar K(X))]_{ij}$, with \((i,j)\) the input-output channel indices, satisfies
    \[
    \operatorname{diag}\!\left(
    [\mathcal F_{\ell}(\mathcal P,\bar K(X))]_{ij}
    \right)
    =
    \operatorname{diag}\!\left(
    [\mathcal F_{\ell}(\mathcal P,K_0)]_{ij}
    \right).
    \]
    \item The interconnection signal $\epsilon$ satisfies \(
\epsilon = e_0,
\)
with \(e_0\) the local error signal of the baseline system.
\end{enumerate}
\end{theorem}

\begin{proof}
Substituting~\eqref{eq:coprime_K} and~\eqref{eq:coprime_G} into~\eqref{eq:youla_lft} 
gives \(N_cD_c^{-1}=K_0\), \(D_c^{-1}=I\), \(D_c^{-1}N_p=Z\), and
\(D_p+N_cD_c^{-1}N_p = \hat G^{-1}Z+K_0Z = (S\hat G)^{-1}Z\),
with \(S=(I+\hat G K_0)^{-1}\), yielding the controller expression 
in~\eqref{eq: lft controller}. To establish the closed-loop transfer function, 
note that for \(P_{22}= -G\),
\[
\begin{aligned}
I-P_{22}\bar K(X)
&= I+ G K_0 - S^{-1}ZX(I+ZX)^{-1} \\
&= S^{-1} - S^{-1}ZX(I+ZX)^{-1} \\
&= S^{-1}\bigl(I-ZX(I+ZX)^{-1}\bigr) \\
&= S^{-1}(I+ZX)^{-1}.
\end{aligned}
\]
and therefore \(\bigl(I-P_{22}\bar K(X)\bigr)^{-1} = (I+ZX)S\). 
Substituting into the LFT interconnection yields the affine form
\[
\mathcal F_{\ell}(\mathcal P,\bar K(X))
=
\bar T_{11} + \bar T_{12} X \bar T_{21},
\]
where \(\bar T_{11} = P_{11}+P_{12}K_0SP_{21} = \mathcal F_{\ell}(\mathcal P, K_0)\), 
\(\bar T_{12} = -P_{12}\hat G^{-1}Z\), and \(\bar T_{21} = SP_{21}\). 
Statement~(i) follows by inspection. Statement~(ii) follows from the Youla 
construction~\citep{Zhou1999EssentialsControl}. Statement~(iii) follows since 
\(\bar T_{12}\), \(\bar T_{21}\) are diagonal and \(X\) is hollow. 
Statement~(iv) follows from \(\epsilon = (I+ZX)^{-1}e\) combined with 
\(e = (I+ZX)e_0\), where \(e_0\) denotes the baseline error. 
\end{proof}
The resulting control interconnection is shown in Figure~\ref{fig:extended_control_structure}, with the closed-loop transfer 
functions summarized in~\eqref{tab: transfers}, where $T = I - S$, 
\(S_{\mathrm{i}} = (I + K_0 G)^{-1}\), and $T_{\mathrm{i}} = I - S_{\mathrm{i}}$. Each transfer function consists of the diagonal baseline closed-loop transfer function augmented by a coupling term affine in \(X\). The add-on controller acts through the factor \((I+ZX)\), reshaping the baseline local error vector before projection by \(W\). Since \(W\) has a nontrivial null space, this freedom can be exploited to align the local error vector with \(\ker(W)\), reducing the components that contribute 
to \(e_{\mathcal S}\). Furthermore, Theorem~\ref{thm:theorem_1} establishes 
that the parameterization satisfies both requirements of Problem~\ref{prob:1}, since setting \(X=0\) directly recovers \(\Sigma^0\), and Statement~(\textit{ii}) guarantees internal stability of \(\Sigma\) for any~\(X\).

\begin{figure}[t]
    \centering
    \includegraphics[width=\linewidth]{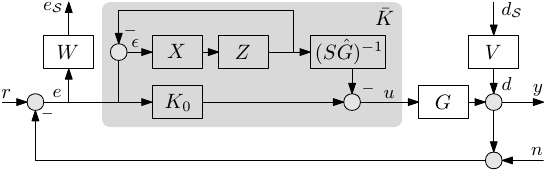}
    \caption{Control interconnection consisting of the baseline controllers \(K_0\) and the add-on system-level controllers.}

    \label{fig:extended_control_structure}
\end{figure}

\begin{figure*}[b!] 
\centering
\begingroup
\setlength{\abovedisplayskip}{6pt}
\setlength{\belowdisplayskip}{6pt}
\begin{equation} \label{tab: transfers}
\left[\begin{array}{c}
e_{\mathcal{S}} \\ e \\ \epsilon \\ y \\ u
\end{array}\right]
=
\left[\begin{array}{c|c|c|c}
-W(I+ZX) S & -W(I+ZX) S G & W(I+ZX) S & -W(I+ZX) S \\
-(I+ZX) S & -(I+ZX) S G & (I+ZX) S & -(I+ZX) S \\
-S & -SG & S & -S \\
(I+ZX) S & (I+ZX) S G & T-ZXS & -(T-ZXS) \\
-(K_0-G^{-1}ZX)S & -T_{\mathrm{i}}+G^{-1}ZXSG & (K_0-G^{-1}ZX)S & -(K_0-G^{-1}ZX)S 
\end{array}\right]
\left[\begin{array}{c}
d_y \\ d_u \\ r \\ n
\end{array}\right]
\end{equation}
\endgroup
\end{figure*}



\subsection{Performance analysis} \label{sec: analysis}
Next, the conditions under which the transfer function matrix from 
disturbances $d$ to the system-level error $e_{\mathcal{S}}$, denoted 
$T_{e_{\mathcal{S}}d} : d \mapsto e_{\mathcal{S}}$, can be set to zero 
using the add-on controller, i.e., $T_{e_{\mathcal{S}}d} = 0$, are 
characterized.  This analysis is conducted under the assumption \(G^{-1} \in \mathcal{RH}^{N \times N}_{\infty}\), which allows \(Z = I\). The following theorem states the conditions on the performance matrix \(W\) under which this is achievable. \\

\begin{theorem} \label{theorem 2}
Let \(G^{-1} \in \mathcal{RH}^{N \times N}_{\infty}\), \(Z = I\) and \(X \in \mathbb{R}^{N\times N}\)  with \(\operatorname{diag}(X)=0\) and $N$ the number of subsystems. There exist static couplings \(X\) such that \(T_{e_{\mathcal{S}} d}(X) = 0\) if and only if there exists a vector $v \in \operatorname{ker}(W)$ with $v_j \neq 0$, $j \in \{1,\ldots,N\}$.
\end{theorem}
\begin{proof}
For $ Z=I$, the condition \(T_{e_{\mathcal{S}} d} = 0\) is equivalent to \(W(I + X) = 0\). This requires each column of $I+X$ to be an element of the null space $\operatorname{ker}(W)$. Since the diagonal of $I+X$ is fixed to one, the $j$th column of $I+X$ can only be an element of $\operatorname{ker}(W)$ if there exists $v \in \operatorname{ker}(W)$ with $v_j \neq 0$. If no such vector exists, the diagonal constraint cannot be satisfied, and no feasible $X$ exists. Conversely, if this condition holds for all \(j\), such an \(X\) can be constructed by choosing each column of \(I + X\) as a vector in \(\operatorname{ker}(W)\) normalized to have a unit entry at position \(j\). By construction, \(W(I + X) = 0\) and \(\operatorname{diag}(X) = 0\), completing the proof.
\end{proof}

\begin{corollary} \label{col 1}
Under the conditions of Theorem~\ref{theorem 2}, achieving 
\(T_{e_{\mathcal{S}} d}(X)  = 0\) imposes \(r_w N\) independent linear constraints on the coupling matrix \(X\). Consequently, since \(X\) contains \(N(N-1)\) free entries, the number of active couplings \(n_X\) required to achieve $T_{e_{\mathcal{S}}d} = 0$ must satisfy
\begin{equation}
    r_w N \leq n_X \leq N(N-1).
\end{equation}
\end{corollary}
Theorem~\ref{theorem 2} characterizes the ideal case in which the plant can be
stably and causally inverted. In this case, \(G^{-1} \in \mathcal{RH}_{\infty}^{N \times N}\), such that \(Z=I\), and a static coupling \(X\) provides sufficient freedom to achieve
\(T_{e_{\mathcal{S}} d}=0\). If \(G^{-1} \notin \mathcal{RH}_{\infty}^{N \times N}\), this direct cancellation is no longer admissible. The filter \(Z\) is then required to ensure \(G^{-1}Z \in \mathcal{RH}_{\infty}^{N \times N}\), and captures the stable and causal inversion limitations imposed by nonminimum-phase zeros and relative-degree constraints. Corollary~\ref{col 1} further shows that only a limited number of interconnections are typically required, implying that the available degrees of freedom can be used either to obtain sparse interconnection structures or address secondary objectives. The null-space condition in Theorem~\ref{theorem 2} implies that \(T_{e_{\mathcal{S}} d}=0\) is achievable only if \(W\) defines system-level objectives that enable redistribution among subsystems, rather than independent local objectives. This is illustrated in the following example. \\

\begin{example}
Consider the performance matrix
\begin{equation} \label{eq: W1}
W =
\begin{bmatrix}
1 & -1 & 0 \\
0 &  0 & 1 
\end{bmatrix},
\end{equation}
with \(\operatorname{ker}(W) = \operatorname{span}\{[1,  1,  0]^\top\}\). Since \(v_3 = 0\), the null space condition of Theorem~\ref{theorem 2} is not satisfied, and no static coupling \(X\) exists such that \(T_{e_{\mathcal{S}} d}=0\). This performance matrix represents a coupled objective for subsystems~\(j=1,2\) and an independent, local objective for subsystem~\(j=3\), which will always remain nonzero in the system-level error $e_{\mathcal{S}}$. In contrast, for
\begin{equation} \label{eq: W2}
W =
\begin{bmatrix}
1 & -1 & 0 & 0\\
0 &  0 & 1 & -1
\end{bmatrix},
\end{equation}
with \(\operatorname{ker}(W) = \operatorname{span}\{[1,  1,  0,  0]^\top, [0,  0,  1,  1]^\top\}\), each coordinate contains a nonzero entry. Hence, the null space condition is satisfied, and a static coupling \(X\) exists that achieves \(T_{e_{\mathcal{S}}d}=0\). In this case, $W$ describes a coupled objective for subsystem $j=1,2$ and $j= 3,4$, which enables the contribution of each subsystem to be redistributed among subsystems.
\end{example}

\section{System-level add-on controller synthesis}
This section develops the synthesis procedure for sparse system-level 
controllers, providing Contribution~C2.

\subsection{Sparse $\mathcal{H}_2$ synthesis problem}
The performance criterion is defined as the squared \(\mathcal{H}_2\) norm
of the closed-loop transfer function,
\begin{equation} \label{eq: cost function}
\mathcal{J}(X) = \bigl\|T_{zw}(X)\bigr\|_{\mathcal{H}_2}^2,
\end{equation}
with \(T_{zw}(X)=\mathcal W\,\mathcal F_{\ell}(\mathcal P,\bar K(X))\mathcal{V}\),
as illustrated in Figure~\ref{fig:lft}. The synthesis problem is formulated as
\begin{equation} \label{eq: problem}
\begin{aligned}
\hat X = \argmin_{X \in \mathbb{R}^{N \times N}} \quad & \mathcal{J}(X)
+ \lambda \|\operatorname{vec}(X)\|_1, \\
\text{s.t.} \quad & \operatorname{diag}(X) = 0,
\end{aligned}
\end{equation}
where \(\lambda > 0\) is the regularization parameter. The \(\ell_1\) regularization is included to reduce the interconnection complexity of \(X\), promoting sparse interconnection structures and selecting the interconnections most effective in improving performance. The problem is thus to compute sparse solutions \(\hat X\) that minimize~\eqref{eq: problem}, given \(\hat G\), \(K_0\), \(W\), \(V\), and~\(\lambda\).

\subsection{LASSO reformulation}
This section derives the synthesis procedure for solving~\eqref{eq: problem}. By Theorem~\ref{thm:theorem_1}, the closed-loop map is affine in \(X\).
Consequently, the closed-loop transfer function \(T_{zw}(X)\) can be written as
\begin{equation} \label{eq:Tzw_1}
T_{zw}(X) = T^0_{zw} + T_{12} X T_{21},
\end{equation}
with
\[
\begin{aligned}
T^0_{zw} &= \mathcal W \mathcal F_\ell(\mathcal P, K_0) \mathcal V, \\
T_{12} &= -\mathcal W P_{12} \hat G^{-1} Z, \\
T_{21} &= S P_{21} \mathcal V,
\end{aligned}
\] 
where $T^0_{zw}$ is the closed-loop transfer function of the baseline system $\Sigma^0$. Vectorizing~\eqref{eq:Tzw_1}, the cost function~\eqref{eq: cost function} becomes
\begin{equation} \label{eq: cost combined}
    \mathcal{J}(X)
    = \left\| l + \mathcal{A} x
    \right\|_{\mathcal{H}_2}^2,
\end{equation}
with $x = \operatorname{vec}(X)$, $l = \operatorname{vec}(T^0_{zw})$, and $\mathcal{A} = T_{21}^\top \otimes T_{12}$. Expanding the $\mathcal{H}_2$ norm yields the quadratic form
\begin{equation} \label{eq: quadratic cost}
\mathcal{J}(X) = x^\top Q x + 2 x^\top f + c,
\end{equation}
where $c$ is a constant and
\begin{align}
Q &= \frac{1}{2\pi} \int_{-\infty}^{\infty}
         \mathcal{A}^\hop(j\omega) \mathcal{A}(j\omega) \, d\omega, \label{eq: Q} \\
f &= \frac{1}{2\pi} \int_{-\infty}^{\infty}
         \mathcal{A}^\hop(j\omega) \, l(j\omega) \, d\omega, \label{eq: f}
\end{align}
with $Q \in \mathbb{R}^{N^2 \times N^2}$ positive definite and $f \in \mathbb{R}^{N^2}$. The diagonal constraint $\operatorname{diag}(X) = 0$ is eliminated by the variable transformation
\begin{equation} \label{eq: param transform}
    x = \Gamma \theta,
\end{equation}
where $\Gamma \in \mathbb{R}^{N^2 \times N(N-1)}$ is a selection matrix that removes the diagonal entries. Substituting into~\eqref{eq: quadratic cost} and applying the Cholesky decomposition $\tilde{Q} = A^\top A$ with $\tilde{Q} = \Gamma^\top Q \Gamma$, $\tilde{f} = \Gamma^\top f$, and $b = -A^{-\top} \tilde{f}$, problem~\eqref{eq: problem} reduces to the standard LASSO form~\citep{Tibshirani1996RegressionLasso}
\begin{equation} \label{eq: problem LARS}
    \hat{\theta}
    = \argmin_{\theta \in \mathbb{R}^{N(N-1)}}
    \|A\theta - b\|_2^2
    + \lambda \|\theta\|_1.
\end{equation}
The solution along the entire regularization path of $\lambda$ is efficiently computed using the LARS algorithm~\citep{Efron2004LeastRegressions}, yielding $\hat{\theta}$ for any specified number of active interconnections $n_X$. A second unregularized least-squares problem is then solved on the support of $\hat{\theta}$ to remove the regularization bias. The complete procedure is summarized in Algorithm~\ref{alg: synthesis}.

\begin{algorithm}[t]
\caption{Sparse add-on controller synthesis}
\label{alg: synthesis}
\begin{algorithmic}[1]
\Require $\hat G$, $K_0$, $W$, $V$, $n_X$
\State Construct $l$ and $\mathcal{A}$ 
\State Compute $Q$ and $f$ as defined in~\eqref{eq: Q} and~\eqref{eq: f}, and apply the parameter transformation given in~\eqref{eq: param transform}.
\State Obtain $A$ and $b$ from the Cholesky factorization of $\tilde{Q}$.
\State Solve~\eqref{eq: problem LARS} using LARS to obtain $\hat{\theta}$ with $n_X$ active parameters
\State De-bias by solving an unregularized LS problem on the support of $\hat{\theta}$
\Ensure sparse coupling $\hat X$
\end{algorithmic}
\end{algorithm}

\section{Simulation results}
The synthesis procedure is evaluated on a system of \(N=18\) locally controlled subsystems. Each local plant is given by \(G_{ii}(s)=1/(M_{ii} s^2)\), where the masses \(M_{ii}\) are randomly selected in the range \(35\)–\(40~\mathrm{kg}\). The subsystems are controlled by fixed PID-type controllers with closed-loop bandwidths between \(25\) and \(75~\mathrm{Hz}\). The disturbance model \(V\) is chosen as a low-pass filter with a cut-off frequency of \(50\,\mathrm{Hz}\), and the system-level performance objective is defined by a randomly generated matrix \(W \in \mathbb{R}^{3 \times 18}\), corresponding to three combined performance objectives. In addition to minimizing the system-level error, the synthesis includes control action as a secondary objective.

Figure~\ref{fig: pareto} shows the normalized cost \(\mathcal{J}(X)/\mathcal{J}_0\) as a function of the number of active interconnections \(n_X\). The results show that most interconnections contribute only marginally to performance improvement, since the cost rapidly approaches its minimum for sparse interconnection matrices. The selected sparsity levels \(n_X \in \{30,45,100\}\) are visualized in Figure ~\ref{fig: adjacency}, which shows the corresponding interconnection structures of $\hat X$. Figure~\ref{fig: frf} compares the Frobenius norm of \(T_{e_{\mathcal S}d}\) for the baseline system, the sparse solutions, and the dense solution. The sparse solutions closely match the dense solution, confirming that only a limited number of critical interconnections are required to achieve near-optimal performance.

\begin{figure*}
    \centering
    \includegraphics[width=0.89\linewidth]{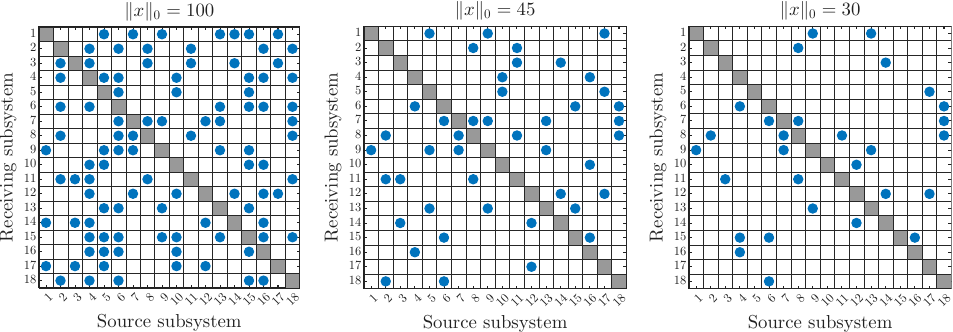} \vspace{-0.5em}
\caption{Sparsity pattern of the optimal coupling matrix \(\hat X\) for \(n_X \in \{100, 45, 30\}\), where each active interconnection \tikzdot{MatlabBlue} corresponds to a nonzero entry \(\hat X_{ij}\) coupling subsystem \(j\) to receiving subsystem~\(i\).}
    \label{fig: adjacency}
\end{figure*}

\begin{figure}[t]
\vspace{-0.75em}
    \centering
    \includegraphics[scale=0.95]{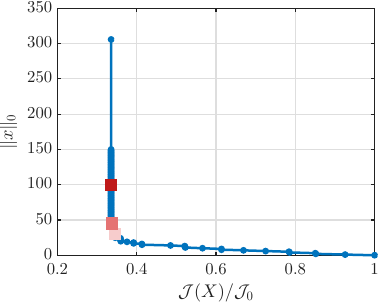} \vspace{-0.5em}
\caption{Pareto trade-off between the normalized cost \(\mathcal{J}(X)/\mathcal{J}_0\) and the number of active interconnections \(n_X \) \tikzdot{MatlabBlue}, showing that most interconnections contribute only marginally to performance. The points \(n_X \in \{100, 45, 30\}\) \tikzdot{MatlabRed} correspond to the solutions in Fig.~\ref{fig: adjacency} and the frequency responses in Fig.~\ref{fig: frf}.}
    \label{fig: pareto}
\end{figure}

\begin{figure}[t]
\vspace{-0.25em}
    \centering
    \includegraphics[scale=0.95]{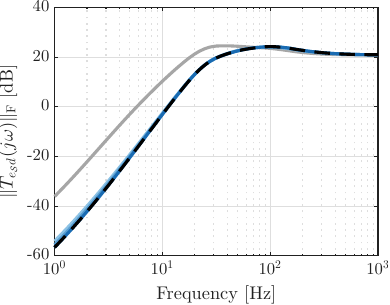} \vspace{-0.5em}
    \caption{Bode magnitude of the Frobenius norm of \(T_{e_{\mathcal S} d}(j\omega)\) for the baseline system without coupling \tikzline{MatlabGray}, sparse system-level controllers with \(n_X \in \{30, 45, 100\}\) \tikzline{MatlabBlue}, and the full interconnection \tikzdashedline{MatlabBlack}, where the sparse and dense interconnections achieve similar performance.}
    \label{fig: frf}
\end{figure}

\section{Conclusion}
This paper introduced a Youla-based add-on control framework that preserves the existing local controllers while enabling system-level performance improvements. The affine dependence of the closed-loop transfer functions on the interconnection matrix \(X\) leads to a convex \(\mathcal{H}_2\) synthesis problem for sparse coupling controllers. The synthesis procedure enables an optimal trade-off between system-level performance and interconnection complexity. Simulation results show that near-optimal performance can be achieved using only a limited number of active interconnections.

\end{document}